\documentclass[]{article}
\usepackage[english]{babel}
\usepackage[autostyle,italian=guillemets]{csquotes}
\usepackage[backend=biber,style=numeric]{biblatex}
\usepackage{amssymb}
\usepackage{mathtools}
\usepackage{amsthm}
\usepackage{hyperref}

\theoremstyle{definition}
\newtheorem{definition}{Definition}[section]

\theoremstyle{plain}

\newtheorem{theorem}{Theorem}
\newtheorem{lemma}{Lemma}
\newtheorem{proposition}{Proposition}

\newcommand{\defeq}{\overset{\underset{\mathrm{def}}{}}{=}}
\newcommand{\N}{\mathbb{N}}
\newcommand{\B}{\mathbb{B}}

\title{On the complexity of SAT}
\author{Fabio Romano}

\begin{document}

\maketitle

\begin{abstract}
In this paper, we prove that no deterministic algorithm can solve SAT in polynomial time in the number of boolean variables.
\end{abstract}

\section{First definitions}

\begin{definition}[Set of input symbols $\Sigma$]
	The set of \emph{input symbols} $\Sigma$ for SAT language is made-up by:
	
	\begin{itemize}
		\item
		The symbols $x, 0, 1$.
		
		\item
		Binary operators $\land$ and $\lor$, standing for the logical AND and OR of two expressions.
		
		\item
		Unary operator $\neg$, standing for logical negation.
		
		\item
		Parentheses to group operators and operands, if necessary to alter the default precedence of operators: $\neg$ highest, then $\land$ and finally $\lor$.
	\end{itemize}
	
	Cf. \cite[sect. 10.2]{hopcroft:computation}.
\end{definition}

\begin{definition}[Boolean variables]
	If $w \in \{0, 1\}^*$ is a string starting with 1, then $xw$ is a boolean variable, where $xw$ is the concatenation of symbol $x$ with $w$. $w$ is the binary representation of the index of variable $xw$.
	
	Cf. \cite[par. 10.2.2]{hopcroft:computation}.
\end{definition}

\begin{definition}[Boolean expressions]
	\mbox{}
	
	\begin{itemize}
		\item
		Every boolean variable is a boolean expression.
		
		\item
		If $A$ and $B$ are boolean expressions, then $(A)$, $\neg A$, 	$A \land B$ and $A \lor B$ are boolean expressions.
	\end{itemize}
	
	Cf. \cite[par. 10.2.1]{hopcroft:computation}.
\end{definition}

\begin{definition}[Length of a boolean expression]
	The length of a boolean expression $E$ is the number of
	positions for symbols in $E$. We denote the length of $E$ with $|E|$.
\end{definition}

\begin{definition}[Set of boolean expressions not longer than $n$] \label{def:bool_exp_leq_n}
	For every $n \in \N$, the set of boolean expressions not longer than $n$ is defined as follows:
	
	\[
	BoolExp_\leq(n) \defeq \{E  \mid E \text{ is a boolean expression s. t. } |E| \leq n\}
	\]
	
	From now on, we assume that $0 \notin \N$.
\end{definition}

\begin{definition}[Truth assignment]
	A \emph{truth assignment} for a given boolean expression $E$ assigns either 1 or 0 to each of the variables mentioned in $E$. The \emph{value} of expression $E$ given a truth assignment $T$, denoted $E(T)$, is the result of evaluating $E$ with each variable $x$ replaced by the value $T(x)$ that $T$ assigns to $x$ (cf. \cite[par. 10.2.1]{hopcroft:computation}).
\end{definition}

\begin{definition}[Satisfiability]
	A truth assignment $T$ \emph{satisfies} a boolean expression $E$ if $E(T) = 1$; i.e., the
	truth assignment $T$ makes expression $E$ true. A boolean expression $E$ is said to be \emph{satisfiable} if there exists at least one truth assignment $T$ that satisfies $E$ (cf. \cite[par. 10.2.1]{hopcroft:computation}).
\end{definition}

\begin{definition}[$\text{SAT} \subseteq \Sigma^*$]
	\emph{SAT} is the set of all and only boolean expressions that are satisfiable (cf. \cite[par. 10.2.1]{hopcroft:computation}).
\end{definition}

\begin{definition}[Language of a dTM]
	Given a deterministic TM $M$, the language of $M$ $L(M)$ is the set of coded input strings accepted by $M$ (cf. \cite[sect. 8.2]{hopcroft:computation}).
\end{definition}

\begin{definition} [Time complexity] \label{def:time_complexity}
	A dTM $M$ is said to be of time complexity $T(w)$ (or to have running time $T(w)$) if whenever $M$ is given an input $w$, $M$ halts after making at most $T(w)$ moves.
\end{definition}

\begin{definition}[The class $\mathcal P$] \label{def:class_P}
	We say a language $L \in \mathcal P$ iff there is some polynomial $P(n)$ such that $L = L(M)$ for some deterministic TM of time complexity $P(n)$ in the length of input (cf. \cite[par. 10.1.1]{hopcroft:computation}).
\end{definition}

\section{Issues with definition of complexity for SAT}

At first sight, it might be thought that $\text{SAT} \in \mathcal P$ is a sufficient condition to affirm that SAT is a tractable problem, in fact, up to now it has been so. However, a previous attempt we made to prove that  $\text{SAT} \notin \mathcal P$ has pointed out some problems in this regard.

The fundamental trouble concerns the concept of time complexity: by definition \ref{def:time_complexity}, the running time of an algorithm is a function $T(w)$ on the input. So, if $\text{SAT} \in \mathcal P$, then, by definition \ref{def:class_P}, we would have a SAT solver that runs in polynomial time $P(|w|)$ in the length of an input $w$. But despite this, we don't know anything about relation between the length of $w$ and the number of its boolean variables.

Since a boolean expression always represents a boolean function, the shortest expression that can represent a boolean function of $n$ variables could virtually have a length $l$ that is exponential in $n$. So, if $l = 2^n$, the SAT solver can run in time $P(2^n)$, which is obviously exponential in the number of variables, although it is polynomial in the length of input.

Therefore, it follows that the definition of $\mathcal P$ is not satisfying enough to ensure that SAT is tractable, so SAT could be exponential even if $\text{SAT} \in \mathcal P$. Later in the paper, we will see that this is how things really are.

For this reason, we have to attack the problem with techniques of boolean function theory.

\section{Boolean Function Theory}

\begin{definition}[Boolean domain]
	$\B \defeq \{0, 1\}$ is the set of truth values.
\end{definition}

\begin{definition}[Boolean function]
	Given an $n \in \N$, a boolean function $f$ is a total function $f \colon \B^n \to \B$, i.e. it assigns to each $n$-tuple of elements of $\B$ an element of $\B$.
\end{definition}

\begin{definition}[Equivalence of boolean functions] \label{def:bool_func_equiv}
	Let $f \colon \B^n \to \B$ and $g \colon \B^m \to \B$ be two boolean functions.
	
	\begin{itemize}
		\item If $n \geq m$:
		
		$f \equiv g$ iff, for every $\langle v_1, \dots, v_n \rangle \in \B^n,\ f(v_1, \dots, v_n) = g(v_1, \dots, v_m)$ holds.
		
		\item If $n < m$:
		
		$f \equiv g$ iff $g \equiv f$.
	\end{itemize}
	
	If $f \equiv g$, we say that $f$ and $g$ are equivalent. It's trivial to prove that $\equiv$ is an equivalence relation.
\end{definition}

\begin{definition}[Interpretation of boolean expressions]
	Given a boolean expression $E$ that contains $x_{k_1}, \dots, x_{k_n}$ distinct variables sorted by index, $I_E \colon \B^{k_n} \to \B$ is a boolean function such that:
	
	\[
	I_E(v) \defeq E(T_v)
	\]
	
	where $v = \langle v_1, \dots, v_{k_n} \rangle \in \B^{k_n}$ and $T_v \colon \{x_{k_1}, \dots, x_{k_n}\} \to \mathcal \B$ is the truth assignment such that, for every $i \in \{1, \dots, n\}$, $T_v(x_{k_i}) = v_{k_i}$.
	
	We say that $I_E$ is the interpretation of $E$ as a boolean function.
\end{definition}

\begin{definition}[Length of a boolean function] \label{def:func_length}
	Given a boolean function $f \colon \mathcal \B^n \to \mathcal \B$, the length of $f$ is defined as follows:
	
	\[
	length(f) \defeq \min\{|E| \mid I_E \equiv f,\ E \text{ is a boolean expression}\}
	\]
	
	That is, the length of a boolean function $f$ is the length of the shortest boolean expression $E$, such that its interpretation is equivalent to $f$.
\end{definition}

\begin{definition}[Set of $n$-ary boolean functions] \label{def:n-ary_bool_func}
	For every $n \in \N$, the set of all $n$-ary boolean functions is defined as follows:
	
	\[
	BoolFunctions(n) \defeq \{f \colon \mathcal \B^n \to \mathcal \B \mid f^{-1}(\B) = \B^n\}
	\]
	
	The condition $f^{-1}(\B) = \B^n$ stands to mean that $f$ is a total function, i.e. the inverse image of $f$ coincides with its domain.
\end{definition}

\section{Proof}

\begin{proposition} \label{prop:bool_exp_subset}
	For every $n \in \N$:
	
	\[
	BoolExp_\leq(n) \subseteq \bigcup_{k = 0}^n \Sigma^k
	\]
\end{proposition}

\begin{proof}
	If $E \in BoolExp_\leq(n)$, then, by def. \ref{def:bool_exp_leq_n}, $E$ is a boolean expression such that $|E| \leq n$, and since $E \in \Sigma^*$, exists a $k \in \{0, \dots, n\}$ such that $E \in \Sigma^k$, so $E \in \bigcup_{k = 0}^n \Sigma^k$.
\end{proof}

\begin{lemma} \label{lemma:card_compare}
	If $P$ is a polynomial, then exists an $N_P \in \N$ such that:
	
	\[
	|BoolExp_\leq(P(N_P))| < |BoolFunctions(N_P)|
	\]
\end{lemma}

\begin{proof}
	Let's consider the cardinality of $BoolExp_\leq(P(n))$:
	\[
		\begin{split}
			|BoolExp_\leq(P(n))|
			&\leq \left| \bigcup_{k = 0}^{P(n)} \Sigma^k \right| \text{ by prop. \ref{prop:bool_exp_subset}} \\
			&= \sum_{k = 0}^{P(n)} |\Sigma|^k \text{ since, for each } i \neq j,\ \Sigma^i \cap \Sigma^j = \emptyset \\
			&= \frac{|\Sigma|^{P(n)+1}-1}{|\Sigma|-1} \text{ since is a geometric series} \\
			&= \frac{8^{P(n)+1}-1}{8-1} \text{ since $\Sigma$ contains 8 elements} \\
			&= \frac{8}{7} \cdot 8^{P(n)} - \frac{1}{7} \\
			&= O(8^{P(n)})
		\end{split}
	\]
	
	So, $|BoolExp_\leq(P(n))| = O(8^{P(n)})$.
	
	Now, let's consider the set $BoolFunctions(n)$. By def. \ref{def:n-ary_bool_func}, $BoolFunctions(n)$ is the set of $n$-ary boolean functions. Let be $f \in BoolFunctions(n)$.
	
	Since $|\B^n|=2^n$, then the cardinality of the domain of $f$ is $2^n$, and since, for each $n$-tuple of $\B^n$, $f$ can assign either 0 or 1, then the number of different ways $f$ can assign boolean values to the $n$-tuples is $2^{|\B^n|} = 2^{2^n}$, so $|BoolFunctions(n)| = 2^{2^n}$.
	
	Now, let's compare the growth of cardinalities of $BoolExp_\leq(P(n))$ and $BoolFunctions(n)$ using their logarithms to base 2, which are roughly the number of bits necessary to represent the cardinalities.
	
	\[
	\log |BoolExp_\leq(P(n))| = O(\log 8^{P(n)}) = O(P(n) \cdot \log 8) = O(P(n))
	\]
	
	\[
	\log |BoolFunctions(n)| = \log 2^{2^n} = 2^n \cdot \log 2 = 2^n
	\]
	
	This means that the number of bits to represent $|BoolFunctions(n)|$ grows exponentially in the arity of boolean functions, while the number of bits to represent $|BoolExp_\leq(P(n))|$ grows polynomially. So, by interpolation, we can always find an $N_P \in \N$ such that:
	
	\[
	|BoolExp_\leq(P(N_P))| < |BoolFunctions(N_P)| \qedhere
	\]
\end{proof}

\begin{proposition} \label{prop:diff_exp_for_diff_func}
	If $f \colon \B^n \to \B$ and $g \colon \B^n \to \B$ are two boolean functions s. t. $f \neq g$, and $F, G$ two boolean expressions s. t. $I_F \equiv f$ and $I_G \equiv g$, then $F \neq G$.
\end{proposition}

\begin{proof}
	If $F = G$, then $I_F = I_G$, so, by equivalence, $I_F \equiv f$ and $I_G \equiv g$ entail $f \equiv g$, and since $f$ and $g$ have the same arity, by def. \ref{def:bool_func_equiv} $f = g$ holds, contradicting $f \neq g$.
\end{proof}

\begin{lemma} \label{lemma:func_length_greater_poly}
	If $P$ is a polynomial, then exists a boolean function $\mathcal F \colon \B^{N_P} \to \B$ such that $P(N_P) < length(\mathcal F)$.
\end{lemma}

\begin{proof}
	By lemma \ref{lemma:card_compare}, there is an $N_P \in \N$ such that:
	
	\begin{equation} \label{eq:card_compare}
		|BoolExp_\leq(P(N_P))| < |BoolFunctions(N_P)|
	\end{equation}
	
	Let be $g, h \in BoolFunctions(N_P)$ s. t. $g \neq h$, and $G, H$ two boolean expressions s. t. $I_G \equiv g$ and $I_H \equiv h$. Since $g$ and $h$ are both $N_P$-ary functions, by prop. \ref{prop:diff_exp_for_diff_func} $G \neq H$.
	
	This entails that we can define a partial function $\mathcal M \colon BoolExp_\leq(P(N_P)) \to BoolFunctions(N_P)$ as follows:
	
	\begin{equation} \label{eq:def_M}
		\mathcal M(E) \defeq f, \text{ such that } I_E \equiv f \in BoolFunctions(N_P)
	\end{equation}
	
	Now, suppose by contradiction that, for every $f \in BoolFunctions(N_P)$, there is an $E_f \in BoolExp_\leq(P(N_P))$ such that $f = \mathcal M(E_f)$. Then, $\mathcal M$ would be surjective, but this involves that $|BoolExp_\leq(P(N_P))| \geq |BoolFunctions(N_P)|$, contradicting \eqref{eq:card_compare}.
	
	This means that $\mathcal M$ cannot be surjective, so, by \eqref{eq:def_M}, exists a boolean function $\mathcal F \colon \B^{N_P} \to \B$ s. t., for every $E \in BoolExp_\leq(P(N_P))$, $I_E \not\equiv \mathcal F$ holds.
	
	By functional completeness of $\{\neg, \land, \lor\}$ and by def. \ref{def:func_length}, there must be a boolean expression $F$ such that $I_F \equiv \mathcal F$ and $|F| = length(\mathcal F)$. Since, for every $E \in BoolExp_\leq(P(N_P))$, $I_E \not\equiv \mathcal F$ holds, then $F \notin BoolExp_\leq(P(N_P))$, so, by def. \ref{def:bool_exp_leq_n}, $P(N_P) < |F|$, thus $|F| = length(\mathcal F)$ entails that $P(N_P) < length(\mathcal F)$.
	
	Summing up, there is a boolean function $\mathcal F \colon \B^{N_P} \to \B$ such that $P(N_P) < length(\mathcal F)$, which is the thesis of this lemma.
\end{proof}

\begin{lemma} \label{lemma:simplification}
	Le be $f \colon \B^n \to \B$ a boolean function, and $F$ a boolean expression that contains $x_{k_1}, \dots, x_{k_m}$ distinct variables sorted by index, s. t. $I_F \equiv f$ and $k_m > n$. Then, exists a boolean expression $F'$ with its variable indexes less or equal to $n$, s. t. $I_{F'} \equiv f$ and $|F'| < |F|$.
\end{lemma}

\begin{proof}
	Since $I_F \equiv f$, $F$ contains $k_m$ variables and $k_m > n$, then, by def. \ref{def:bool_func_equiv}, for every $\langle v_1, \dots, v_n \rangle \in \B^n$:
	\[
	f(v_1, \dots, v_n) = I_F(v_1, \dots, v_n, \underbrace{0, \dots, 0}_{k_m - n \text{ times}})
	\]
	
	Then, by def. \ref{def:bool_func_equiv}, $I_F \equiv f$ entails that, for every $\langle v_1, \dots, v_{k_m} \rangle \in \B^{k_m}$:
	
	\[
	I_F(v_1, \dots, v_{k_m}) = I_F(v_1, \dots, v_n, \underbrace{0, \dots, 0}_{k_m - n \text{ times}})
	\]
	
	This means the truth of $F$ does not depend by the values of $x_{n+1}, \dots, x_{k_m}$, so we can replace them in $F$ with 0, and then simplify $F$ until we obtain a new boolean expression, with the rules below:
	
	\[
	\frac{E} {E[\neg 1 / 0]}
	\qquad
	\frac{E} {E[\neg 0 / 1]}
	\]
	
	\[
	\frac{E} {E[A \land 1 / A]}
	\qquad
	\frac{E} {E[1 \land A / A]}
	\qquad
	\frac{E} {E[A \lor 1 / 1]}
	\qquad
	\frac{E} {E[1 \lor A / 1]}
	\]
	
	\[
	\frac{E} {E[A \land 0 / 0]}
	\qquad
	\frac{E} {E[0 \land A / 0]}
	\qquad
	\frac{E} {E[A \lor 0 / A]}
	\quad
	\frac{E} {E[0 \lor A / A]}
	\]
	
	where $E[A / B]$ is the result of replacing in $E$ every occurrence of $A$ with $B$. If simplification results in a boolean value, then $1$ can be replaced by $x_1 \lor \neg x_1$ and $0$ by $x_1 \land \neg x_1$.
	
	The result of simplifying $F$ is clearly a boolean expression $F'$ with its variable indexes less or equal to $n$, s. t. $I_{F'} \equiv f$ and $|F'| < |F|$.
\end{proof}

\begin{theorem}
	There is no dTM $M$, with polynomial time complexity in the number of boolean variables of input, such that $L(M) = \text{SAT}$.
\end{theorem}

\begin{proof}
	Suppose by contradiction that exists a dTM $M$, with polynomial time complexity $P(n)$ in the number of boolean variables of input, such that $L(M) = \text{SAT}$. Then, by lemma \ref{lemma:func_length_greater_poly}, there is a boolean function $f \colon \B^{N_P} \to \B$ such that $P(N_P) < length(f)$.
	
	By functional completeness of $\{\neg, \land, \lor\}$ and by def. \ref{def:func_length}, there is a boolean expression $F$ such that $I_F \equiv f$ and $|F| = length(f)$, so $P(N_P) < length(f)$ entails that $P(N_P) < |F|$.
	
	Since, by def. \ref{def:func_length}, $F$ is the shortest boolean expression such that $I_F \equiv f$, and since $f$ is an $N_P$-ary function, then, by lemma \ref{lemma:simplification}, $F$ contains at most $N_P$ variables, and since $P$ is a polynomial, $P(N_P)$ is certainly an upper bound of the running time of $M$ with $F$ as input.
	
	But, on the other hand, to verify that $F$ is a boolean expression, $M$ is forced to read the entire input, so $M$ must make at least $|F|$ moves, therefore $P(N_P) \geq |F|$, contradicting $P(N_P) < |F|$.
\end{proof}

\section{Conclusions}

Although we didn't prove that $\text{SAT} \notin \mathcal P$, anyway we understand an important fact: even if SAT were in $\mathcal P$, over a certain $n \in \N$, the quantity of $n$-ary boolean functions that are writable with polynomial length is infimum, if compared to the totality of $n$-ary boolean functions. So, even in that case, a minimal part of problem instances would be tractable, while the overwhelming majority would still be effectively unsolvable in polynomial time (in the number of variables).

\printbibliography[heading=bibintoc]

\end{document}